\documentclass[letterpaper, 10 pt, conference]{ieeeconf}  

\IEEEoverridecommandlockouts                              

\usepackage{amssymb}
\usepackage{amsmath}
\usepackage{siunitx}
\usepackage{graphicx}
\usepackage{tikz}
\usepackage{pgfplots}
\pgfplotsset{compat=1.16}

\usepackage{booktabs}
\usepackage{multirow}
\usepackage{color, colortbl}
\usepackage{hyperref}

\usepackage{pdfpages}
\usepackage{balance}

\usepackage{import}

\usetikzlibrary{plotmarks}
\usetikzlibrary{arrows.meta}
\usepgfplotslibrary{patchplots}
\usepackage{grffile}

\title{\LARGE \bf
RLO-MPC: Robust Learning-Based Output Feedback MPC for Improving the Performance of Uncertain Systems in Iterative Tasks 
}

\author{Lukas Brunke, Siqi Zhou, and Angela P. Schoellig
\thanks{The authors are with the Dynamic Systems Lab
	(http://www.dynsyslab.org), Institute for Aerospace Studies,
	University of Toronto, Canada. The authors are also affiliated with
	the Vector Institute for Artificial Intelligence, Toronto. Emails:
	\{lukas.brunke, siqi.zhou\}@robotics.utias.utoronto.ca, schoellig@utias.utoronto.ca}%
}

\renewcommand\vec{\mathbf} 
\newcommand{\set}[1]{\mathbb{#1}}
\newcommand{\R}{\set{R}}

\newcommand{\N}{\set{N}}

\newtheorem{theorem}{Theorem}

\newtheorem{definition}{Definition}

\newtheorem{assumption}{Assumption}
\newtheorem{corollary}{Corollary}

\begin{document}

\maketitle
\thispagestyle{empty}
\pagestyle{empty}

\begin{abstract}
In this work we address the problem of performing a repetitive task when we have uncertain observations and dynamics. 
We formulate this problem as an iterative infinite horizon optimal control problem with output feedback. Previously, this problem was solved for linear time-invariant~(LTI) system for the case when noisy full-state measurements are available using a robust iterative learning control framework, which we refer to as robust learning-based model predictive control~(RL-MPC). 
However, this work does not apply to the case when only noisy observations of part of the state are available. This limits the applicability of current approaches in practice: First, in practical applications we typically do not have access to the full state. Second, uncertainties in the observations, when not accounted for, can lead to instability and constraint violations. To overcome these limitations, we propose a combination of RL-MPC with robust output feedback model predictive control, named robust learning-based output feedback model predictive control~(RLO-MPC). 
We show recursive feasibility and stability, and prove theoretical guarantees on the performance over iterations. We validate the proposed approach with a numerical example in simulation and a quadrotor stabilization task in experiments.
\end{abstract}

\section{INTRODUCTION}
Performing iterative tasks when having uncertain observations and a uncertain dynamics model is a common problem in practice. Such a problem setup arises in industrial robotics~\cite{Arimoto1984}, mobile robotics~\cite{Mckinnon2019}, or in autonomous car racing~\cite{Rosolia_racing}. The uncertainties can be due to sensor noise or nonlinear effects that are difficult to identify and model.

Iterative learning control (ILC) can be applied in such cases to achieve high performance in repetitive tasks. ILC uses data from previous iterations to improve controller performance, for example, by reference signal adaptation for future iterations~\cite{Schoellig2012}.
Model predictive control (MPC) is another popular control strategy. At every time step, MPC solves an open-loop finite horizon optimal control problem and only applies the first control input of the optimal input sequence to the system. This is executed in a receding horizon fashion. The benefit of MPC is that its optimization problem explicitly considers state and input constraints.
Combinations of ILC and linear model predictive control for repetitive tasks have been proposed in the past. In~\cite{Lee1999}, the authors use MPC with a time-varying multiple-input-multiple-output~(MIMO) linear model and use stored tracking errors from past iterations to correct for model errors. The authors extend this work by proving asymptotic tracking for linear constrained systems in~\cite{Lee2000}. In~\cite{Cueli2008}, the authors consider a nonlinear model, linearize the dynamics model in the MPC based on the previous iteration's trajectory, and prove stability and fast convergence to set points. 

\begin{figure}[t]
	\centering
	\vspace{3mm}
	\def\svgwidth{\linewidth}
	\includegraphics[width=\linewidth]{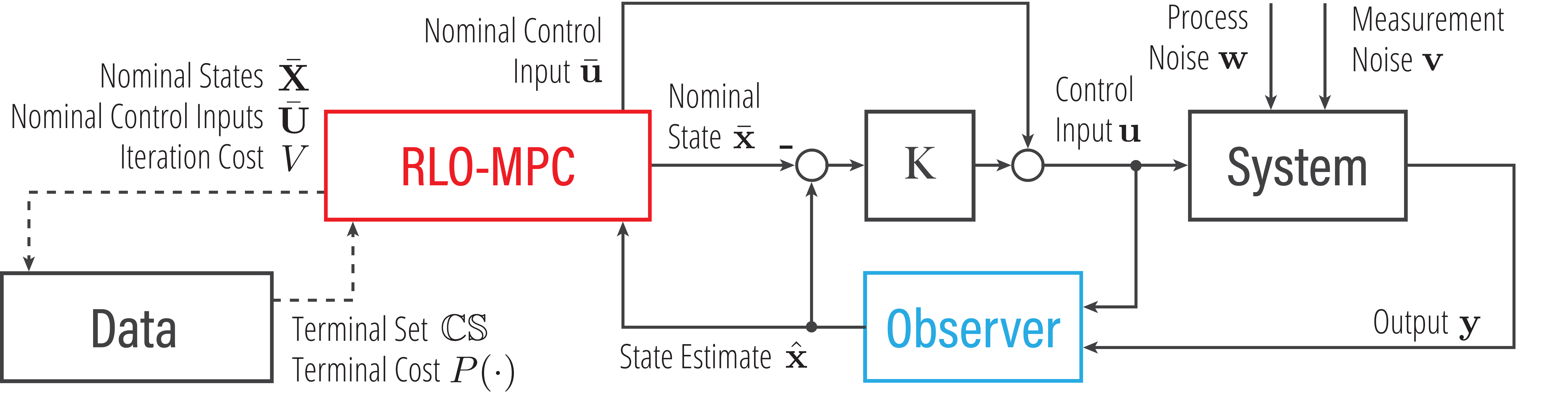}
	\caption{Our proposed robust learning-based output feedback MPC (RLO-MPC, red) combines robust output feedback MPC with iterative learning control.  In contrast to robust output feedback MPC, we propose to determine the terminal set and terminal cost function from data after each iteration (indicated by dashed lines). To handle uncertain outputs, we extend robust learning-based MPC (RL-MPC) with an observer (blue) and account for state estimation errors. } 
	\label{fig:block_diagram}
\end{figure}

Recently, the authors of~\cite{Rosolia2018a} presented an efficient learning-based model predictive control~(L-MPC) framework for a stabilization task. In~\cite{Rosolia2018a}, the terminal cost and terminal constraint set are determined from data of past iterations with the goal to guarantee stability and constraint satisfaction of the MPC. Formal guarantees have been derived for special cases such as LTI systems~\cite{Rosolia2017} and robust learning-based model predictive control~(RL-MPC) for LTI systems with bounded process noise~\cite{Rosolia2018, Rosolia2021}. 

While existing work on L-MPC considers full-state feedback, 
\textcolor{black}{safe learning-based control approaches with output feedback remain to be an open challenge~\cite{DSL2021}}.
In most practical settings, the full state cannot be measured as adequate sensors are not available. State estimation aims to recover the full-state information and its associated uncertainty from partial state observations, called measurements or output. 
The state estimation errors can be accounted for by extending robust MPC methods to also consider
bounded measurement noise and a state estimator. This is referred to as robust output feedback MPC and includes tube-based approaches~\cite{Kogel2017, Lorenzetti2020, Mayne2006} and minmax approaches~\cite{Copp2017, Findeisen2004, Lofberg2002}. In this work, we use a tube-based approach to achieve a lower online computational effort. 

Our proposed robust learning-based output feedback MPC, which we call RLO-MPC, leverages tube-based robust output feedback MPC for guaranteeing stability and constraint satisfaction despite uncertainties and iterative learning for improving performance. 
We extend the RL-MPC for an LTI system with noisy state measurements in~\cite{Rosolia2018} to an LTI system with uncertain outputs. 
We tighten the constraints according to the additional state estimation uncertainty using methods from robust output feedback MPC. We compute the terminal constraint set and the terminal cost using data from previous iterations similar to RL-MPC. 
We derive theoretical guarantees for stability, recursive feasibility, and performance. 
In contrast to previous work~\cite{Rosolia2018, Rosolia2021}, we assume that the initial state is contained inside a compact set. 
While performance improvement over the iterations cannot be guaranteed for this case, we derive upper bounds on the total iteration cost. This upper bound is non-increasing and is reduced over iterations.
The proposed RLO-MPC is validated in simulation and through quadrotor experiments in Sec.~$\ref{sec:experiments}$. In simulation we also show the benefit of the proposed controller as compared to a certainty equivalent approach, which does not consider state estimation errors. 

\section{PROBLEM STATEMENT}

We consider an uncertain discrete-time LTI system:
\begin{align}
\label{eq:system}
\begin{split}
\vec{x}_{t + 1} &= \vec{A} \vec{x}_t + \vec{B} \vec{u}_t + \vec{w}_t \,, \\
\vec{y}_t &= \vec{C} \vec{x}_t + \vec{v}_t \,,
\end{split}
\end{align}
where $\vec{x}_t \in \R^{n_x}$ is the state at time $t$, $\vec{u}_t \in \R^{n_u}$ is the input, $\vec{y}_t \in \R^{n_y}$ is the output and, $\vec{A}$, $\vec{B}$, and $\vec{C}$ are known real system matrices of consistent dimension. The system is assumed to be controllable and observable. The process noise $\vec{w}_t \in \set{W} \subset \R^{n_x}$ and measurement noise $\vec{v}_t \in \set{V} \subset \R^{n_y}$ affect the system at every time step. The sets $\set{W}$ and $\set{V}$ are known compact polytopes that contain the origin in their non-empty interior. The system is subject to known hard state and input constraints 
$\vec{x} \in \set{X}\,, \vec{u} \in \set{U}$,
where $\set{X} \subset \R^{n_x}$ is closed and $\set{U} \subset \R^{n_u}$ is compact. We assume that both $\set{X}$ and $\set{U}$ are polyhedral and contain the origin in their non-empty interior. 

The goal of this work is to design a controller that solves the iterative infinite horizon optimal control problem at every iteration $j$ optimized over $\vec{u}_{k}^j$ for all $k\geq 0$: 
\begin{align}
	\label{eq:olmpc_ftocp}
	\raisetag{60pt}
	\begin{split}
	V^{j, *}_{0 \to \infty}(\vec{y}_S) =
	 \min_{\substack{\vec{u}_{0}^j, \vec{u}_{1}^j, \dots 
	}} 
	& \quad \sum_{k = 0}^{\infty} \left(\vec{y}_{k}^{j} \right)^\intercal \tilde{\vec{Q}} \vec{y}_{k}^j + \left(\vec{u}_{k}^{j} \right)^\intercal  \tilde{\vec{R}} \vec{u}_{k}^{j} \\
	\mathrm{s.t.} & \quad \vec{x}_{k + 1}^j =\vec{A}\vec{x}_{k}^j + \vec{B}\vec{u}_{k}^j + \vec{w}_k^j\,, \\
	& \quad \vec{y}_{k}^j =\vec{C}\vec{x}_{k}^j + \vec{v}_{k}^j\,, \vec{x}_{k}^j \in \set{X} \,, \vec{u}_{k}^j \in \set{U} \,,  \\
	& \quad  \vec{w}_{k}^j \in \set{W} \,, \vec{v}_{k}^j \in \set{V} \,, \vec{y}_{0}^j = \vec{y}_S  \,, \forall k\in \N \,, \\
	\end{split}
\end{align}
where $\tilde{\vec{Q}}$ and $\tilde{\vec{R}}$ are positive definite matrices, $\vec{u}_{k}^{j}$ are output feedback policies, and $\vec{y}_S$ is the initial output. This iterative infinite horizon problem is generally intractable due to the infinite horizon. Typically the infinite horizon is approximated by a finite horizon and a terminal cost that accounts for the cost-to-go. In this paper, we approach this problem by exploiting the iterative nature of the task and by learning the terminal cost function and terminal constraint set for a robust output feedback MPC from previous iterations. 

\section{BACKGROUND}
In this section, we present the necessary background for the proposed approach.

\subsection{Notation}
We denote the set of natural numbers as $\N$ with $0 \in \N$. The set of integers from $i \in \N$ to $j \in \N$ is denoted as $\set{I}_{i,j} = \{i, \dots, j\}$, where $j > i$. Let $\set{A}, \set{B} \subset \R^n$, then the Minkowski sum is defined as $\set{A} \oplus \set{B} = \{\vec{a} + \vec{b} : \vec{a} \in \set{A}, \vec{b} \in \set{B}\}$ and the Pontryagin difference is defined as $\set{A} \ominus \set{B} = \{\vec{x} \in \R^n : \vec{x} \oplus \set{B} \subseteq \set{A}\}$.

\subsection{Uncertainty Bounding with a Single Tube}
The system under consideration is uncertain (see~\eqref{eq:system}). To account for the uncertainties, we design a single error tube according to~\cite{Kogel2017, Lorenzetti2020, Mayne2006} using robust positively invariant (RPI) sets. The error tube guarantees that any errors starting inside the error tube, will stay inside the error tube for all future time steps $t \geq 0$.
\begin{definition}[Robust Positively Invariant (RPI) Set~\cite{Borrelli2017}]
	A set $\set{P} \subseteq \set{X}$ is robust positively invariant for an autonomous system subject to state constraints $\set{X}$, input constraints $\set{U}$, and bounded process noise set $\set{W}$, if the initial system state $\vec{x}_0 \in \set{P}$ implies $\vec{x}_t \in \set{P}\,, \forall \vec{w} \in \set{W}, t \geq 0$.
\end{definition}

We start by introducing the observer and the error state dynamics. Since only the outputs are available, we employ a Luenberger observer to estimate the state:
\begin{align}
\label{eq:luenberger}
	\begin{split}
	\hat{\vec{x}}_{t + 1} &= \vec{A} \hat{\vec{x}}_{t} + \vec{B} \vec{u}_t + \vec{L} (\hat{\vec{y}}_t - \vec{y}_t) \,, \\
	\hat{\vec{y}}_t &= \vec{C} \hat{\vec{x}}_t \,,	
	\end{split}
\end{align}
where $\hat{\vec{x}}_t$ and $\hat{\vec{y}}_t$ are the estimated state and output at time step $t$, and $\vec{L} \in \R^{n_x \times n_y}$ is the static observer gain and is chosen such that $\vec{A}_L = \vec{A} + \vec{L}\vec{C}$ is Schur stable. The resulting dynamics of the state estimation error $\vec{e}_{t} = \vec{x}_t - \hat{\vec{x}}_t$ is
\begin{equation}
\label{eq:estimation_error}
\vec{e}_{t + 1} = \vec{A}_L \vec{e}_{t} + \vec{w}_t + \vec{L} \vec{v}_t \,.
\end{equation}
We additionally introduce the nominal system
\begin{equation}
\label{eq:nominal_system}
\bar{\vec{x}}_{t + 1} = \vec{A} \bar{\vec{x}}_{t} + \vec{B} \bar{\vec{u}}_t \,,
\end{equation}
where $\bar{\vec{x}}_t \in \R^{n_x}$ is the nominal state at time step $t$ and $\bar{\vec{u}}_t \in \R^{n_u}$ is the input to the nominal system. The actual control input to the system then combines the nominal control input as a feedforward term and a feedback term, shown in Figure~\ref{fig:block_diagram}, as
\begin{equation*}
\vec{u}_t = \bar{\vec{u}}_t + \vec{K} \vec{\xi}_t \,,
\end{equation*}
where $\vec{\xi}_t = \hat{\vec{x}}_t - \bar{\vec{x}}_t$ and the static feedback gain $\vec{K} \in \R^{n_u \times n_x}$, which is selected such that $\vec{A}_K = \vec{A} + \vec{B} \vec{K}$ is Schur stable. The error between the estimated and the nominal state yields the dynamics
\begin{equation}
\label{eq:est_nom_error}
\vec{\xi}_{t + 1} = \vec{A}_K \vec{\xi}_t - (\vec{L} \vec{C} \vec{e}_{t} + \vec{L} \vec{v}_t)\,.
\end{equation}
Using the error definitions the system state can be written as
\begin{equation*}
\vec{x}_t = \bar{\vec{x}}_t + \vec{\xi}_t + \vec{e}_{t} \,. 
\end{equation*}
The errors $\vec{\xi}_t$ and $\vec{e}_{t}$ are shown to be bounded in~\cite{Mayne2006}. The estimation error dynamics~\eqref{eq:estimation_error} guarantees the existence of the minimally RPI (mRPI) set $\set{E}_\infty$ that satisfies
\begin{equation}
\label{eq:e_inf}
\set{E}_\infty = (\vec{A} + \vec{LC}) \set{E}_\infty \oplus \set{W} \oplus \vec{L} \set{V} \,.
\end{equation} 
Similarly, the dynamics of the difference between the estimated state and the nominal state~\eqref{eq:est_nom_error} yields the mRPI set $\Xi_\infty$, that satisfies
\begin{equation}
\label{eq:xi_inf}
\Xi_\infty = (\vec{A} + \vec{BK}) \Xi_\infty \oplus (- \vec{LC}) \set{E}_\infty \oplus (- \vec{L}) \set{V} \,.
\end{equation}
The estimation error set~\eqref{eq:e_inf}~and the prediction error set~\eqref{eq:xi_inf} allow us to choose the initial nominal state $\bar{\vec{x}}_0$ and the sequence of nominal control inputs $\bar{\vec{U}} = \begin{pmatrix}
	\bar{\vec{u}}_0, \dots, \bar{\vec{u}}_t, \dots 
\end{pmatrix}$ such that the state and the input of the uncertain system~\eqref{eq:system} satisfy the original constraints for all time steps. 
For reduced conservatism, we follow~\cite{Kogel2017} and combine the error states into a single error system:
\begin{equation}
\label{eq:extended_error_system}
\vec{z}_{t + 1} = \begin{pmatrix}
\vec{e}_{t + 1} \\
\vec{\xi}_{t + 1}
\end{pmatrix} = \vec{F} \vec{z}_t + \vec{G} \vec{d}_t \,,
\end{equation}
where $\vec{d}_t \textcolor{black}{= \begin{pmatrix}
	\vec{w}_t^\intercal, \vec{v}_t^\intercal
\end{pmatrix}^\intercal} \in \set{D} = \set{W} \times \set{V}$ and 
\begin{equation*}
\vec{F} = \begin{pmatrix}
\vec{A}_L & \vec{0} \\
- \vec{L C} & \vec{A}_K
\end{pmatrix} \,, ~
\vec{G} = \begin{pmatrix}
\vec{I} & \vec{L} \\
\vec{0} & - \vec{L}
\end{pmatrix} \,,
\end{equation*}
where $\vec{F}$ is Schur stable by design. It follows that the combined error dynamics are bounded given that $\vec{F}$ is Schur stable and $\set{D}$ is bounded. 
Similar to~\cite{Lorenzetti2020} we consider a convex, compact mRPI set $\set{Z}_\infty$ for the extended error system~\eqref{eq:extended_error_system}:
\begin{equation}
\label{eq:z_inf}
\set{Z}_{\infty} = \vec{F} \set{Z}_\infty \oplus \vec{G} \set{D} \,.
\end{equation}
We can use the set $\set{Z}_\infty$ to tighten the state and inputs constraints as follows:
\begin{equation*}
\bar{\set{X}} = \set{X} \ominus \begin{pmatrix}
\vec{I} & \vec{I} 
\end{pmatrix} \set{Z}_\infty \,, ~ \bar{\set{U}} = \set{U} \ominus \begin{pmatrix}
\vec{0} & \vec{K}
\end{pmatrix} \set{Z}_\infty \,.
\end{equation*}
Compared to~\cite{Mayne2006},~\eqref{eq:z_inf} yields a \textit{single} tube and not two tubes based on~\eqref{eq:e_inf}~and~\eqref{eq:xi_inf}, respectively. This is desirable for tube-based MPC as the state and input constraints are less restrictive.  

\subsection{Robust Output Feedback MPC}
In this subsection, we introduce robust output feedback MPC~\cite{Mayne2006}. In the following discussion, we use $\vec{x}_{k \lvert t}$ to denote the open-loop prediction of state~$\vec{x}_k$ at time step~$t$ and $t_1:t_2$ abbreviates consecutive time indices. Consider the finite time optimal control problem~(FTOCP) initialized at $t = 0$ with $\bar{\vec{x}}_0 = \hat{\vec{x}}_0$: 
\begin{align}
	\label{eq:rmpc_ftocp}
	\raisetag{40pt}
	\begin{split}
	V_{t \to t + H}(\textcolor{black}{\bar{\vec{x}}_t}) =
	\min_{\substack{\bar{\vec{u}}_{t:t + H -1 \vert t}}}
	& \quad \sum_{k = t}^{t + H - 1} \ell\left(\bar{\vec{x}}_{k \vert t}, \bar{\vec{u}}_{k \vert t}\right) + V_f\left(\bar{\vec{x}}_{t + H \vert t}\right) \\
	\text{s.t.} & \quad \forall k\in \set{I}_{t, t+ H - 1} \,, \\
	& \quad \bar{\vec{x}}_{k + 1 \vert t} =\vec{A}\bar{\vec{x}}_{k \vert t} + \vec{B}\bar{\vec{u}}_{k \vert t}\,, \\
	& \quad \bar{\vec{x}}_{k \vert t} \in \bar{\set{X}} \,, \bar{\vec{u}}_{k \vert t} \in \bar{\set{U}} \,,  \\
	& \quad \textcolor{black}{\bar{\vec{x}}_{t \vert t} = \bar{\vec{x}}_{t} } \,, \bar{\vec{x}}_{t + H \vert t} \in \bar{\set{X}}_{f} \,, 
	\end{split}
\end{align}
where $H > 0$ is the horizon,
$\ell\left(\vec{x}, \vec{u}\right) = \vec{x}^\intercal \vec{Q} \vec{x} + \vec{u}^\intercal \vec{R} \vec{u}$,
is the stage cost with positive definite matrices $\vec{Q}$ and $\vec{R}$,
$V_f\left(\vec{x}\right) = \vec{x}^\intercal \vec{P} \vec{x}$, 
is the cost-to-go, where $\vec{P}$ is the solution to the discrete-time algebraic Riccati equation of the unconstrained LQR problem for the nominal system~\eqref{eq:nominal_system}.
The set $\bar{\set{X}}_{f}$ is the terminal constraint set\textcolor{black}{ that is control invariant for the nominal system \eqref{eq:nominal_system} and} determined using the tightened state and input constraints $\bar{\set{X}}$ and $\bar{\set{U}}$ and the closed-loop system matrix $\vec{A} + \vec{B} \vec{K}_f$, where $\vec{K}_f$ is the optimal LQR state feedback gain. 
Only the first optimal input $\bar{\vec{u}}^*_{t \vert t}$ from the FTOCP is used at time step~$t$:
\begin{equation*}
\vec{u}_t^* = \bar{\vec{u}}^*_{t \vert t} + \vec{K} \vec{\xi}_t \,,
\end{equation*}
and applied to the system in \eqref{eq:system}. The FTOCP \textcolor{black}{is} solved in a receding horizon fashion at every time step \textcolor{black}{$t \geq 0$ using~\eqref{eq:rmpc_ftocp}.} 

We consider a quadratic cost and a linear system with polytopic \textcolor{black}{state and input} constraints. \textcolor{black}{The FTOCP in~\eqref{eq:rmpc_ftocp} is a quadratic program (QP) that can be efficiently solved.}
In contrast to~\cite{Kogel2017, Mayne2006}, we do not optimize over the initial state and keep it fixed for all time steps. This allows the comparison of the total iteration cost for varying initial conditions. 

\subsection{Convex Safe Set}
We collect the nominal state and nominal control input for iteration $j$ using
$
\bar{\vec{X}}^j = \begin{pmatrix}
\bar{\vec{x}}_0^j, \dots, \bar{\vec{x}}_t^j, \dots 
\end{pmatrix} \,, \bar{\vec{U}}^j = \begin{pmatrix}
\bar{\vec{u}}_0^j, \dots, \bar{\vec{u}}_t^j, \dots 
\end{pmatrix}$.
As in~\cite{Rosolia2018, Rosolia2017} we define the set of indices of successful iterations, which converge to the origin:
\begin{equation*}
\set{M}^j = \left\{ k \in \set{I}_{0, j} : \lim_{t \to \infty} \bar{\vec{x}}_t^k = \vec{0} \right\} \,.
\end{equation*}
The sampled safe set $\set{SS}^j$ for iteration $j$ is the collection of all nominal states for every successful iteration: 
\begin{equation*}
\set{SS}^j = \left\{ \bigcup_{i \in \set{M}^j} \bigcup_{t = 0}^{\infty} \bar{\vec{x}}_t^j \right\} \,.
\end{equation*} 
Furthermore, we introduce the convex safe set $\set{CS}^j$ with 
\begin{equation}
\label{eq:conv_safe_set}
\set{CS}^j = \mathrm{conv}(\set{SS}^j) \,,
\end{equation}
where $\mathrm{conv}(\cdot)$ indicates the convex hull. For an LTI \textcolor{black}{nominal} system, the convex safe set is a control invariant set and is substituted for the terminal constraint set $\bar{\set{X}}_{f}$ in the robust output feedback MPC~\eqref{eq:rmpc_ftocp}~\cite{Borrelli2017}.

\subsection{Terminal Cost}
The nominal cost-to-go at time step $t$ of iteration $j$ for previously collected state and control input data is given by
\begin{equation*}
V_{t \to \infty}^j (\bar{\vec{x}}_t^j) = \sum_{k = t}^\infty \ell \left(\bar{\vec{x}}_k^j, \bar{\vec{u}}_k^j\right) \,,
\end{equation*}
and the total nominal iteration cost for iteration $j$ is $V_{0 \to \infty}^j(\bar{\vec{x}}_0^j)$. 
Finally, we introduce the barycentric function from~\cite{Rosolia2017} and define it based on the nominal state as in~\textcolor{black}{\cite{Rosolia2021}}:
\begin{align}
	\label{eq:cost_to_go}
	\raisetag{40pt}
	\begin{split}
	P^j(\bar{\vec{x}}) := \min_{\substack{\lambda_t \geq 0, t \in \set{N}_0}} & \quad \sum_{k = 0}^{j} \sum_{t = 0}^{\infty} \lambda_t^k V_{t \to \infty}^k (\bar{\vec{x}}_t^k)  \\
	\mathrm{s.t.} & \quad \sum_{k = 0}^{j} \sum_{t = 0}^{\infty} \lambda_t^k = 1 \,, \sum_{k = 0}^{j} \sum_{t = 0}^{\infty} \lambda_t^k \bar{\vec{x}}_t^k = \bar{\vec{x}} \,. 
	\end{split} 
\end{align}
The function $P^j(\bar{\vec{x}})$ assigns the minimum cost-to-go using a barycentric interpolation to each nominal state inside of the convex safe set $\set{CS}$.
This will replace the terminal cost in the robust output feedback MPC~\eqref{eq:rmpc_ftocp}.

\section{ROBUST LEARNING-BASED OUTPUT FEEDBACK MPC}
\label{sec:roflmpc}
This section introduces the proposed RLO-MPC. We start by presenting the algorithm in Sec.~\ref{sec:roflmpc_algo} and derive the properties of RLO-MPC in Sec.~\ref{sec:roflmpc_properties}.
\subsection{Robust Learning-Based Output Feedback MPC Algorithm}
\label{sec:roflmpc_algo}
The RLO-MPC approximates the infinite horizon problem~\eqref{eq:olmpc_ftocp} as a finite-time horizon optimization problem, which is solved at each time step $t$ of iteration \textcolor{black}{$j > 0$}:
\begin{align}
	\label{eq:roflmpc_ftocp}
	\raisetag{40pt}
	\begin{split}
	V^{j}_{t \to t + H}(\textcolor{black}{\bar{\vec{x}}_t^j}) =  \\ 
	\min_{\substack{\bar{\vec{u}}_{t:t + H -1 \vert t}^j}}
	& \sum_{k = t}^{t + H - 1} \ell\left(\bar{\vec{x}}_{k \vert t}^j, \bar{\vec{u}}_{k \vert t}^j\right) + P^{j - 1}\left(\bar{\vec{x}}_{t + H \vert t}^j\right)  \\
	\mathrm{s.t.} & \quad \forall k\in \set{I}_{t, t + H - 1} \,, \textcolor{black}{\bar{\vec{x}}_{t \vert t}^j = \bar{\vec{x}}_t^j} \,,\\
	& \quad \bar{\vec{x}}_{k + 1 \vert t}^j =\vec{A}\bar{\vec{x}}_{k \vert t}^j + \vec{B}\bar{\vec{u}}_{k \vert t}^j\,, \\
	& \quad \bar{\vec{x}}_{k \vert t}^j \in \bar{\set{X}} \,, \bar{\vec{u}}_{k \vert t}^j \in \bar{\set{U}} \,, \bar{\vec{x}}_{t + H \vert t}^j \in \set{CS}^{j - 1} \,,  \\
	\end{split}
\end{align}
where \textcolor{black}{$\bar{\vec{x}}_0^j = \hat{\vec{x}}_0^j \in \set{E}_0 \subset \bar{\set{X}}$} and $\set{E}_0$ is a compact polyhedral set. The set $\set{E}_0$ accounts for the fact that restarting the next iteration from the same initial state estimate is restrictive and not practical. 
Compared to the output feedback robust MPC \eqref{eq:rmpc_ftocp}, we use $P^{j}$ and $\set{CS}^{j}$ for the terminal cost function and the terminal constraint set, respectively. This leads to an updated terminal cost function and terminal constraint set at every iteration. 
The problem is solved in a receding horizon fashion and the control input to system~\eqref{eq:system} at time step $t$ is given by
\begin{equation}
\label{eq:roflmpc_control_policy}
\vec{u}_t^{j\textcolor{black}{, *}} = \bar{\vec{u}}_{t \vert t}^{j, *} + \vec{K} \vec{\xi}_t^j \,,
\end{equation}
where $\vec{\xi}_t^j = \hat{\vec{x}}_t^j - \bar{\vec{x}}_t^j$, $\bar{\vec{u}}_{t \vert t}^*$ is the input at time $t$ of the optimal solution to~\eqref{eq:roflmpc_ftocp} for the nominal system~\eqref{eq:nominal_system}. The RLO-MPC is initialized with $P^{0}$ and $\set{CS}^{0}$, which is a given feasible, suboptimal solution to the FTOCP. As in~\cite{Rosolia2018}, we make the following assumption: 
\begin{assumption}
	\label{as:safe_set_init}
	We are given an initial nominal state and input sequence $\bar{\vec{X}}^0$, $\bar{\vec{U}}^0$, that satisfy the nominal system dynamics \eqref{eq:nominal_system} and the tightened constraints $\bar{\set{X}}$ and $\bar{\set{U}}$, the initial state estimation constraint $\set{E}_0$, and successfully drive the nominal system to the origin. 
	\textcolor{black}{Further, the initial state estimation error satisfies $\vec{e}_0^j \in \set{E}_\infty, \forall j \geq 0$. }
\end{assumption} 

Assumption~\ref{as:safe_set_init} ensures that the terminal set at the first iteration is non-empty. This assumption is not restrictive in practice, as a conservative controller tuning or a human operator can provide a suboptimal trajectory for most systems. \textcolor{black}{One approach to reduce the conservatism is to optimize the initial nominal state at every time step as in~\cite[Ch.~8]{RosoliaThesis}. However, this would require solving a quadratically constrained quadratic program, which increases the online computation cost. 
}

\subsection{Properties of RLO-MPC}
\label{sec:roflmpc_properties}
In the following, we present the properties of the proposed RLO-MPC. 
To guarantee feasibility for every future iteration, we assume: 
\textcolor{black}{
	\begin{assumption}
		\label{as:reachable_safe_set}
		There exists a feasible solution to~\eqref{eq:roflmpc_ftocp} at $j = 1$ for all $\hat{\vec{x}}_0^1 \in \set{E}_0$. 
	\end{assumption}
}
%
%
%
%

\subsubsection{Recursive Feasibility and Stability}
We present the following corollary on the recursive feasibility and stability of the proposed RLO-MPC:
\begin{corollary}
	Consider system~\eqref{eq:system},~\eqref{eq:luenberger} in closed-loop with the RLO-MPC controller~\eqref{eq:roflmpc_ftocp}~and~\eqref{eq:roflmpc_control_policy}. Under Assumptions~\ref{as:safe_set_init}~and~\ref{as:reachable_safe_set}, the RLO-MPC~\eqref{eq:roflmpc_ftocp},~\eqref{eq:roflmpc_control_policy} is feasible for all times $t \geq 0$ and all iterations $j> 0$. Moreover, the system's state $\vec{x}_t^j$ asymptotically converges to the set $\set{\begin{pmatrix}
		\vec{I} & \vec{I} 
		\end{pmatrix} \set{Z}_\infty}$ for every iteration $j>0$, all $\vec{w} \in \set{W}$, and all $\vec{v} \in \set{V}$.
\end{corollary}
\begin{proof}
	Under Assumptions~\ref{as:safe_set_init}~and~\ref{as:reachable_safe_set}, recursive feasibility and stability for every iteration $j$ follows from Theorem~1 in~\cite{Rosolia2018} and Proposition~2 in~\cite{Lorenzetti2020}.
\end{proof}

\subsubsection{Bound on Total Iteration Cost}
Next we present the results on the controller performance and derive upper bounds on the total nominal iteration cost.
We start by analyzing the performance over the different iterations $j$. 
We generalize the results in~\cite{Rosolia2018} by considering a set of initial states over different iterations $j$. For the case of a fixed initial condition (i.e., $\set{E}_0 = \{\hat{\vec{x}}_0^0\}$), we recover the result in~\cite{Rosolia2018} stating that the total iteration cost is  non-increasing. 
However, the total iteration cost is not guaranteed to be non-increasing for when the initial state changes from iteration to iteration. For linear systems with quadratic cost, the sublevel set of the optimal total iteration cost contains all initial conditions with less or equal total iteration cost. Any initial condition outside of this sublevel set results in an increased total iteration cost. 

In the following, we introduce a non-increasing upper bound on the total iteration cost by considering the set of all possible initial state estimates. 
Before presenting this upper bound, we introduce the  
greatest commonly $T$-step reachable set, which we define in terms of the $T$-step reachable set $\R_T(\cdot)$ \textcolor{black}{for the nominal system with tightened constraints}~\cite{Borrelli2017}.
%
%
\textcolor{black}{%
\begin{definition}[Greatest Commonly $T$-Step Reachable Set]%
	Let $\set{S} \subset \R^n$ be a compact polytopic set, $T \in \set{N}$, and $\mathrm{extreme}(\set{S})$ the set of vertices of $\set{S}$. Then under linear dynamics and polytopic state and input constraints the set 
	\begin{equation*}
		\bar{\R}_T(\set{S}) = \bigcap_{\vec{s} \in \mathrm{extreme}(\set{S})} \R_T (\vec{s}) \,,
	\end{equation*}
	is the greatest commonly $T$-step reachable set of $\set{S}$.
\end{definition}%
}%
\textcolor{black}{Intuitively, the greatest commonly $T$-step reachable set of a set $\set{S}$ is the maximal set of states that the system can reach after $T$ steps if starting from \textcolor{black}{all} $\vec{s} \in \set{S}$. Then Assumption~\ref{as:reachable_safe_set} guarantees that there exists $T \in \{1, \dots, H\}$ such that the intersection $\set{CS}^{0} \cap \bar{\set{R}}_T (\set{E}_{0})$ is non-empty. 
}%

The upper bound on the total cost for iteration $j > 0$ is
\begin{equation}
\label{eq:general_upper_bound}
\bar{V}_{0 \to \infty}^j(\set{E}_0) = 
\min_{\substack{\bar{\vec{x}} \in \set{CS}^{j - 1} \cap \bar{\set{R}}_T (\set{E}_0)}} 
V^{j, *}_{0 \to T}(\set{E}_0, \bar{\vec{x}}) + 
P^{j- 1} (\bar{\vec{x}}) \,,
\end{equation}
where $V^{j, *}_{0 \to T}(\set{E}_0, \bar{\vec{x}})$ is the solution to the following FTOCP:
\begin{align}
\label{eq:ftocp_upper_bound}
\raisetag{60pt}
\begin{split}
V^{j, *}_{0 \to T}(\set{E}_0, \bar{\vec{x}}) = 
\max_{\substack{\bar{\vec{x}}_{0 \vert 0}^j}}
\min_{\substack{\bar{\vec{u}}_{t:t + T -1 \vert 0}^j}}
 & \quad \sum_{k = 0}^{T - 1} \ell(\bar{\vec{x}}_{k \vert 0}^j, \bar{\vec{u}}_{k \vert 0}^{j})  \\
\mathrm{s.t.} & \quad \forall k \in \set{I}_{0, T-1} \,,  \\
& \quad \bar{\vec{x}}_{k + 1 \vert 0}^j =\vec{A}\bar{\vec{x}}_{k \vert 0}^j + \vec{B}\bar{\vec{u}}_{k \vert 0}^j\,, \\
& \quad \bar{\vec{x}}_{k \vert 0}^j \in \bar{\set{X}} \,, \bar{\vec{u}}_{k \vert 0}^j \in \bar{\set{U}} \,,  \\
& \quad \bar{\vec{x}}_{0 \vert 0}^j \in \set{E}_0\,, \bar{\vec{x}}_{T \vert 0}^j = \bar{\vec{x}} \,. 
\end{split}
\end{align}
The objective function~\eqref{eq:general_upper_bound} for $\bar{V}_{0 \to \infty}^j(\set{E}_0)$ is the sum of two optimal costs. The first term gives the worst-case nominal cost for any possible initial state estimate in the set $\set{E}_0$ to any feasible state in $\set{CS}^{j - 1}$. 
\textcolor{black}{The optimization problem in~\eqref{eq:ftocp_upper_bound} is feasible by Assumption~\ref{as:reachable_safe_set} and can be solved by applying line searches to find the initial nominal state on the boundary of $\set{E}_0$.}
The second term is the cost-to-go from the optimal intermediate state $\bar{\vec{x}}^*$ to the origin of the nominal system.  
\begin{theorem}
	\label{th:general_bound}
	Consider system~\eqref{eq:system},~\eqref{eq:luenberger} in closed-loop with the RLO-MPC controller~\eqref{eq:roflmpc_ftocp}~and~\eqref{eq:roflmpc_control_policy}. Let $\set{CS}^j$ be the convex safe set at the $j$-th iteration as defined in~\eqref{eq:conv_safe_set}. Under Assumptions~\ref{as:safe_set_init}~and~\ref{as:reachable_safe_set}, $\bar{V}_{0 \to \infty}^j(\set{E}_0)$ is an upper bound on the total iteration cost for iteration $j > 0$:
	\begin{equation}
	V_{0 \to \infty}^j(\hat{\vec{x}}_0^j) \leq \bar{V}_{0 \to \infty}^j(\set{E}_0) \,.
	\end{equation}
	Furthermore, the upper bound on the total iteration cost $\bar{V}_{0 \to \infty}^j$ is non-increasing:
	\begin{equation}
	\bar{V}_{0 \to \infty}^{j + 1}(\set{E}_0) \leq \bar{V}_{0 \to \infty}^j(\set{E}_0) \,, \forall j > 0  \,.
	\end{equation}
\end{theorem}
\begin{proof}
	We use the fact that the combination of the FTOCP solution for $V^{j, *}_{0 \to T}(\set{E}_0, \bar{\vec{x}})$ and the iteration cost from any previous iteration for $P^{j- 1} (\bar{\vec{x}})$ is suboptimal. Optimality holds only if $\bar{\vec{x}}_T^j = \bar{\vec{x}}^i_t$ \textcolor{black}{for any $i \in \set{I}_{0, j - 1}$}. This yields 
	\begin{align}
	\begin{split}
	\label{eq:ineq_bound_1}
	\raisetag{30pt}
	V_{0 \to \infty}^j(\hat{\vec{x}}_0^j) &\leq  V_{0 \to T}^j(\hat{\vec{x}}_0^j) + 
	P^j(\bar{\vec{x}}_T^j)  \\
	& \textcolor{black}{\leq  V_{0 \to T}^{j, *}(\hat{\vec{x}}_0^j, \bar{\vec{x}}_T^j) + 	P^j(\bar{\vec{x}}_T^j)}  \\
	&\leq \min_{\substack{\bar{\vec{x}} \in \set{CS}^{j - 1} \cap \set{R}_T (\hat{\vec{x}}_0^j)}} V^{j, *}_{0 \to T}(\hat{\vec{x}}_0^j, \bar{\vec{x}}) + P^{j - 1}(\bar{\vec{x}}) \,.
	\end{split}
	\end{align}
	In~\eqref{eq:ineq_bound_1} the total iteration cost is bounded using the initial state estimate $\hat{\vec{x}}_0^j$. An upper bound on the total iteration cost can be obtained by maximizing the cost over all possible initial state estimates in the set $\set{E}_0$:
	\begin{align}
	\begin{split}
	\label{eq:ineq_bound_2}
	\raisetag{20pt}
	V_{0 \to \infty}^j(\hat{\vec{x}}_0^j) &\leq \min_{\bar{\vec{x}} \in \set{CS}^{j - 1} \cap \bar{\set{R}}_T (\set{E}_0)} V^{j, *}_{0 \to T}(\set{E}_0, \bar{\vec{x}}) + P^{j - 1}(\bar{\vec{x}})  \\
	&= \bar{V}_{0 \to \infty}^j(\set{E}_0) \,.
	\end{split}
	\end{align}
	To show $\bar{V}_{0 \to \infty}^j(\set{E}_0)$ is non-increasing over all iterations $j > 0$,
	we first note that the definition of the convex safe set yields 
	$\set{CS}^j \subseteq \set{CS}^{j + 1}$.
	By Assumptions~\ref{as:safe_set_init}~and~\ref{as:reachable_safe_set}, the FTOCP problem~\eqref{eq:roflmpc_ftocp} is feasible at $j = 0$. It follows that 
	\begin{equation}
	\label{eq:recursive_intersection}
	\set{CS}^{j} \cap \bar{\set{R}}_T (\set{E}_0) \subseteq \set{CS}^{j + 1} \cap \bar{\set{R}}_T (\set{E}_0) \,,
	\end{equation}
	which implies that the problem is feasible for all subsequent iterations $j > 0$.
	Since the FTOCP in~\eqref{eq:ftocp_upper_bound} \textcolor{black}{includes a maximization over a quadratic function} with polytopic constraints on the initial state, the solution $\bar{\vec{x}}_0^{j, *}$ at iteration $j$ will lie on the boundary of the polytope $\set{E}_0$, such that $\bar{\vec{x}}_0^{j, *} \in \partial \set{E}_0$. 
	If the iteration $j$ yields an optimal intermediate state $\bar{\vec{x}}^*$ that decreases the total iteration cost, the upper bound on the total iteration will decrease as well. Otherwise the upper bound on the total iteration cost will be constant. Then~\eqref{eq:recursive_intersection},~\eqref{eq:ftocp_upper_bound} and the definition of $P^j$~\eqref{eq:cost_to_go} yield:
	\begin{align*}
	\begin{split}
	\bar{V}_{0 \to \infty}^{j + 1}(\set{E}_0) &= \min_{\substack{\bar{\vec{x}} \in \set{CS}^{j} \cap \bar{\set{R}}_T (\set{E}_0)}} V^{j + 1, *}_{0 \to T}(\set{E}_0, \bar{\vec{x}}) + P^{j} (\bar{\vec{x}})  \\
	&\leq \min_{\substack{\bar{\vec{x}} \in \set{CS}^{j - 1} \cap \bar{\set{R}}_T (\set{E}_0)}} V^{j, *}_{0 \to T}(\set{E}_0, \bar{\vec{x}}) + P^{j-1} (\bar{\vec{x}})  \\
	&= \bar{V}_{0 \to \infty}^{j}(\set{E}_0)\,,
	\end{split}
	\end{align*}
	which gives the desired result.
\end{proof}

In the proof of Theorem~\ref{th:general_bound}, we are applying two inequalities~\eqref{eq:ineq_bound_1}~and~\eqref{eq:ineq_bound_2}, which indicates that we can state a second, tighter upper bound using only the first inequality. While this bound is tighter, it is not guaranteed to be non-increasing and requires knowledge of the new initial state estimate. Define
\begin{equation}
\label{eq:lower_upper_bound}
\underbar{V}_{0 \to \infty}^j(\hat{\vec{x}}_0^j) = \min_{\substack{\bar{\vec{x}} \in \set{CS}^{j - 1} \cap \set{R}_T (\hat{\vec{x}}_0^j)}} V^{j, *}_{0 \to T}(\hat{\vec{x}}_0^j, \bar{\vec{x}}) + P^{j - 1} (\bar{\vec{x}}) \,.
\end{equation}
Then we can state the following Corollary:
\begin{corollary}
	\label{co:lower_upper_bound}
	Consider system~\eqref{eq:system} in closed-loop with the RLO-MPC controller~\eqref{eq:roflmpc_ftocp}~and~\eqref{eq:roflmpc_control_policy}. Let $\set{CS}^j$ be the convex safe set at the $j$-th iteration as defined in~\eqref{eq:conv_safe_set}. Let Assumptions~\ref{as:safe_set_init}~and~\ref{as:reachable_safe_set} hold, then $\underbar{V}_{0 \to \infty}^j(\hat{\vec{x}}_0^j)$ is an upper bound on the total iteration cost for iteration $j > 0$:
	\begin{equation*}
	V_{0 \to \infty}^j(\hat{\vec{x}}_0^j) \leq \underbar{V}_{0 \to \infty}^j(\hat{\vec{x}}_0^j) \,.
	\end{equation*}
\end{corollary}
\begin{proof}
	The proof of Corollary~\ref{co:lower_upper_bound} follows from~\eqref{eq:ineq_bound_1} in the proof of Theorem~\ref{th:general_bound}.
\end{proof}
The results from Theorem~\ref{th:general_bound} and Corollary~\ref{co:lower_upper_bound} are connected through the following relationship:
\begin{equation*}
V_{0 \to \infty}^j(\hat{\vec{x}}_0^j) \leq \underbar{V}_{0 \to \infty}^j(\hat{\vec{x}}_0^j) \leq \bar{V}_{0 \to \infty}^j(\set{E}_0) \,.
\end{equation*}

\section{ILLUSTRATIVE EXAMPLES}
\label{sec:experiments}
This section shows examples for the proposed RLO-MPC presented in Sec.~\ref{sec:roflmpc}.
\subsection{Simulation Example}
Consider the following double integrator system used in~\cite{Mayne2006}:
\begin{align*}
	\begin{split}
	\vec{x}_{k + 1} &=\begin{pmatrix}
	1 & 1 \\
	0 & 1 
	\end{pmatrix}\vec{x}_{k} + \begin{pmatrix}
	1 \\
	1
	\end{pmatrix}\vec{u}_{k} + \vec{w}_k \\
	\vec{y}_{k} &=\begin{pmatrix}
	1 & 1
	\end{pmatrix}\vec{x}_{k} + \vec{v}_{k} \,,
	\end{split}
\end{align*}
where $\set{X} = \{\begin{pmatrix}
x_1 & x_2
\end{pmatrix}^\intercal \in \R^2 : x_1 \in \left[-50, 3\right], x_2 \in \left[-50, 3\right] \}$, $\set{U} = \{ u \in \R : u \in \left[-3, 3\right] \}$, $\set{W} = \{\vec{w} \in \R^2 : \lVert \vec{w} \rVert_\infty \leq 0.1 \}$, and 
$\set{V} = \{v \in \R : \lvert v \rvert \leq 0.05 \}$. We pick $\set{E}_0 = (-8.7~ -0.5)^\intercal \oplus \{\vec{e} \in \R^2 : \lVert \vec{e} \rVert_\infty \leq 0.2 \}$ as the tolerance for the initial state estimate for each iteration. 
For all experiments the RLO-MPC is implemented with a horizon of $H = 2$, the controller gain $\vec{K} = \begin{pmatrix}
-1 & -1
\end{pmatrix}$ and the observer gain $\vec{L} = \begin{pmatrix}
-1 & -1
\end{pmatrix}^\intercal$. We run the controller for $N = 20$ time steps for 30 iterations \textcolor{black}{for computational feasibility but this is} sufficient for the nominal system to reach the origin. For the cost function of \eqref{eq:roflmpc_ftocp} we have $\vec{Q} = \vec{I}$ and $\vec{R} = 0.01$. 
For this example, the minimally robust positively invariant sets $\set{E}_\infty$, $\Xi_\infty$, $\set{Z}_\infty$ can be determined exactly, because the matrices $\vec{A}_L$, $\vec{A}_K$, and $\vec{F}$ are nilpotent~\cite{Rakovic2005}. 
The set $\set{CS}^0$ is initialized by using the robust output feedback MPC approach by~\cite{Kogel2017}, with horizon $H_{\mathrm{init}} = 5$, $N_{\mathrm{init}} = 50$, $\vec{Q}_{\mathrm{init}} = \mathrm{diag}(0.1, 0.1)$, $\vec{R}_{\mathrm{init}} =10$. The terminal cost function $V_f(\vec{x}) = \vec{x}^\intercal \vec{P} \vec{x}$ and the terminal constraint set $\bar{\set{X}}_{f}$ are calculated using $\vec{P}$ and $\vec{K}_f$ from the associated unconstrained LQR problem for the nominal system~\cite{Kouvaritakis2015}. The choice of $\vec{Q}_{\mathrm{init}}$ and $\vec{R}_{\mathrm{init}}$ is such that the performance improvement of the robust output feedback controller is clearly visible. The initial $P^{0} (\cdot)$ is determined with the actual weight matrices $\vec{Q}$ and $\vec{R}$. The experiments are implemented in MATLAB using MPT3~\cite{Herceg2013} and YALMIP~\cite{Lofberg2004} with the solvers Gurobi~\cite{gurobi} and BMIBNB.

\subsubsection{Bounded Total Iteration Cost}
\label{sec:exp_bounds}
We validate the upper bounds on the total iteration cost through simulation. We run the controller 50 times with each run having 30 iterations. An additional superscript $r \in \{1, \dots, 50\}$ indicates the index of the run. Each run and iteration is initialized with a random initial state estimate $\hat{\vec{x}}_0^{j, r} \in \set{E}_0$ using rejection sampling. During this experiment, we assume that the uncertainties $\vec{w}_k^{j, r} \in \mathrm{extreme(\set{W})}$ and $\vec{v}_k^{j, r} \in \mathrm{extreme(\set{V})}$, such that the maximum admissible uncertainties are applied to the system. We pick $T = H$ for this set of simulations. 

The total nominal iteration cost for all runs and the minimum non-increasing upper bound over all runs are shown in Figure~\ref{fig:upper_bounds_boxplot}. As expected, the nominal total iteration cost $V^{j, r}_{0 \to N}(\hat{\vec{x}}_0^{j, r})$ is not guaranteed to be non-increasing due to the variation of initial states across different iterations. However, all iterations of the RLO-MPC have a lower total nominal iteration cost than the initial trajectory for each run. The graph confirms that the relationship $V^j_{0 \to N}(\hat{\vec{x}}_0^{j, r}) \leq \bar{V}^j_{0 \to N}(\set{E}_0)$ holds. The upper bound $\bar{V}^j_{0 \to N}(\set{E}_0)$ is conservative, but the simulation validates that it is non-increasing. This upper bound also shows that the total nominal iteration cost is guaranteed to be lower than the total nominal iteration cost for the initial trajectory from the second iteration onwards. 
Figure~\ref{fig:lower_bound_boxplot} shows the iteration cost prediction error $\Delta V^{j, r}_{0 \to N}(\hat{\vec{x}}_0^{j, r}) = \underbar{V}^{j, r}_{0 \to N}(\hat{\vec{x}}_0^{j, r}) - V^{j, r}_{0 \to N}(\hat{\vec{x}}_0^j)$, which is the difference between the tighter upper bound~\eqref{eq:lower_upper_bound} and the nominal total iteration cost. The error is always nonnegative, which verifies $V_{0 \to \infty}^j(\hat{\vec{x}}_0^{j, r}) \leq \underbar{V}_{0 \to \infty}^j(\hat{\vec{x}}_0^{j, r})$. From iteration~8 onwards the median of the iteration cost prediction error $\Delta V^j_{0 \to N}(\hat{\vec{x}}_0^{j, r})$ over all runs is approximately $0$. Therefore, the tighter upper bound of the total iteration cost $\underbar{V}^{j, r}_{0 \to N}(\hat{\vec{x}}_0^{j, r})$ is indistinguishable from the true total nominal iteration cost $V^j_{0 \to N}(\hat{\vec{x}}_0^{j, r})$ for most runs. 
The tighter upper bound allows us to accurately determine the closed-loop performance of the proposed RLO-MPC for an increasing iteration count.

\begin{figure}[tb]
	\centering
	\input{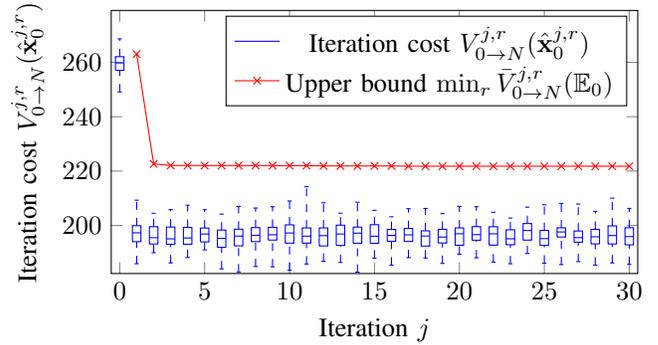}
	\caption{The total nominal iteration cost $V^{j, r}_{0 \to N}(\hat{\vec{x}}_0^{j, r})$ and the derived upper bound $\bar{V}^j_{0 \to N}(\set{E}_0)$ for iterations $j \in \{0, \dots, 30\}$ and runs $r \in \{1, \dots, 50\}$.
	We find that the total nominal iteration cost $V^{j, r}_{0 \to N}(\hat{\vec{x}}_0^{j, r})$ is not guaranteed to be non-increasing in this setting. The graph validates the proposed upper bound of the total nominal iteration cost $\bar{V}^j_{0 \to N}(\set{E}_0)$. Although this upper bound is conservative, it is non-increasing. } 
	\label{fig:upper_bounds_boxplot}
\end{figure}

\begin{figure}[tb]
	\centering
	\input{lower_bound_boxplot.tex}
	\caption{The iteration cost prediction error $\Delta V^j_{0 \to N}(\hat{\vec{x}}_0^{j, r}) = \underbar{V}^j_{0 \to N}(\hat{\vec{x}}_0^{j, r}) - V^j_{0 \to N}(\hat{\vec{x}}_0^{j, r})$ for iterations $j \in \{0, \dots, 30\}$ and runs $r \in \{1, \dots, 50\}$. The graph validates that the tighter upper bound $\underbar{V}^j_{0 \to N}(\hat{\vec{x}}_0^{j, r})$ is indeed an upper bound. Furthermore, the median of the prediction error converges to~$0$. Therefore, the tighter upper bound allows us to accurately determine the performance of the closed-loop system.} 
	\label{fig:lower_bound_boxplot}
\end{figure}

\subsubsection{Comparison to Certainty Equivalence Approach}
We now compare the performance of an RL-MPC and the proposed RLO-MPC. We consider a system, where full-state feedback is not available and an observer has to be used. Since the RL-MPC requires full-state feedback, we apply the certainty equivalence principle. We assume that the state estimate is equal to the true state with $\hat{\vec{x}}_t^j = \vec{x}_t^j$. The state estimate is again given by the Luenberger observer~\eqref{eq:luenberger}. For the RL-MPC we consider the robust tube-based MPC approach described in~\cite{Mayne2005}, where we keep the initial state fixed. Then the state and control input constraints are tightened using the minimally robust positive invariant set $\set{S}_\infty$ that satisfies $\set{S}_\infty = \vec{A}_K \set{S}_\infty \oplus \set{W}$. As in Sec.~\ref{sec:exp_bounds}, we run the controller for 50 different runs with 30 iterations. Furthermore, this experiment uses the same initial states, initial state estimates, and the same process and measurement noise as applied to the system in the experiment in Sec.~\ref{sec:exp_bounds}. 

In Figure~\ref{fig:violating_constraints_histo} we show the number of state constraint violations for the certainty equivalent approach for each run. Every run has a minimum of 5 and a maximum of 15 constraint violations. On average, almost 10 out of 30 iterations lead to constraint violations for the certainty equivalent approach. An unsuccessful trajectory of the certainty equivalent RL-MPC is displayed in Figure~\ref{fig:violating_constraints}. This trajectory from the third iteration of the first run violates the state constraint at step 2. We highlight that the constraint is violated, although the tube is contained inside $\set{X}$. Therefore, not accounting for state estimation errors in RL-MPC can lead to constraint violation. In contrast, the proposed RLO-MPC satisfies all state and input constraints for every iteration in all runs. 


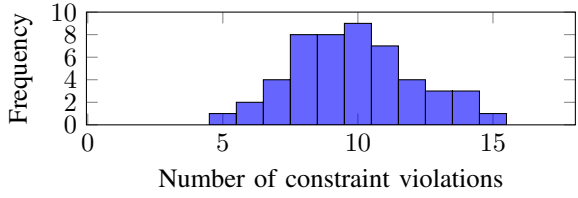
\begin{figure}[tb]
	\centering
%
%
\definecolor{mycolor1}{rgb}{0.00000,0.0,1.0}%
\begin{tikzpicture}

\begin{axis}[%
width=6.5cm,
height=1.5cm,
at={(0.758in,0.481in)},
scale only axis,
xmin=-0.05,
xmax=18.05,
ymin=0,
ymax=10,
xlabel={Number of constraint violations},
ylabel={Frequency},
axis background/.style={fill=white}
]
\addplot[ybar interval, fill=blue, fill opacity=0.6, draw=black, area legend] table[row sep=crcr] {%
x	y\\
4.5	1\\
5.5	2\\
6.5	4\\
7.5	8\\
8.5	8\\
9.5	9\\
10.5	7\\
11.5	4\\
12.5	3\\
13.5	3\\
14.5	1\\
15.5	1\\
};
\end{axis}

\end{tikzpicture}%
	\caption{The number of constraint violations of one experiment run of the certainty equivalent RL-MPC summed over all 30 iterations. The proposed RLO-MPC does not violate any constraints for any run.} 
	\label{fig:violating_constraints_histo}
\end{figure}

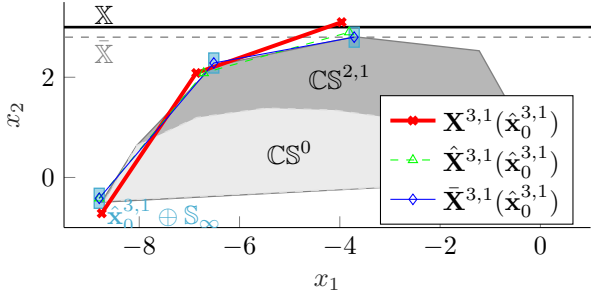
\begin{figure}[tb]
	\centering
%
%
\definecolor{mycolor1}{rgb}{0.00000,1.00000,1.00000}%
\begin{tikzpicture}

\begin{axis}[%
width=7.0cm,
height=3.0cm,
at={(0.758in,0.481in)},
scale only axis,
xmin=-9.5,
xmax=1,
xlabel style={font=\color{white!15!black}},
xlabel={$x_1$},
ymin=-1,
ymax=3.5,
ylabel style={font=\color{white!15!black}},
ylabel={$x_2$},
axis background/.style={fill=white},
axis x line*=bottom,
axis y line*=left,
legend style={at={(0.6,0.3)},anchor=west},
]

\addplot[area legend, line width=0.5pt, draw=gray, fill=white!25.0!gray!75, fill opacity=1.0, forget plot]
table[row sep=crcr] {%
x	y\\
0.157855405218363	0.0703156766799538\\
0.180875604844815	0.0230201996264512\\
0.174969103356666	-0.00590650148814831\\
0.153450106363749	-0.021518996992917\\
0.125411997831009	-0.0280381085327406\\
-8.7551744268349	-0.494237561865138\\
-8.04950659688924	0.639250902588891\\
-6.54941198884056	2.20576243799434\\
-6.42853594175251	2.24657896929688\\
-3.74941198888767	2.79999999995289\\
-3.62853594176444	2.79999999998807\\
-1.22072258378171	2.52868940510596\\
}--cycle;

\addplot[area legend, dashed, line width=0.5pt, draw=white!50.0!gray!50, fill=white!60.0!gray!40.0, fill opacity=1.0, forget plot]
table[row sep=crcr] {%
x	y\\
-8.68875749947814	-0.472070392796104\\
-8.04950659688924	0.639250902588891\\
-6.85334990136685	1.1961566955224\\
-5.46924529271196	1.38410460865489\\
-4.12372388105031	1.34552141166165\\
-2.93984867941637	1.18387520163394\\
-1.96983523177781	0.970013447638557\\
-1.22072258378631	0.749112647991502\\
-0.673474614100328	0.547247969685984\\
-0.296408479628056	0.377066134472273\\
-0.0540434946271233	0.242364985000933\\
0.0875397285384096	0.141583223165533\\
0.157855405218363	0.0703156766799538\\
0.180875604844815	0.0230201996264512\\
0.174969103356666	-0.00590650148814831\\
0.153450106363749	-0.021518996992917\\
0.125411997831009	-0.0280381085327406\\
}--cycle;

\addplot[area legend, line width=1.0pt, draw=black, fill=red, fill opacity=0, forget plot]
table[row sep=crcr] {%
x	y\\
3	3\\
3	-50\\
-50	-50\\
-50	3\\
}--cycle;

\addplot[area legend, line width=0.5pt, dashed, draw=gray, fill=red, fill opacity=0, forget plot]
table[row sep=crcr] {%
x	y\\
2.9	-49.8000000000002\\
-49.9000000000002	-49.8000000000002\\
-49.9000000000002	2.8\\
2.9	2.8\\
}--cycle;

\addplot[area legend, line width=0.5pt, draw=black!25.0!cyan!75, fill=black!25.0!cyan!75, fill opacity=0.5, forget plot]
table[row sep=crcr] {%
x	y\\
-8.69792673273328	-0.61551412806366\\
-8.89792673273328	-0.61551412806366\\
-8.89792673273328	-0.21551412806366\\
-8.69792673273328	-0.21551412806366\\
}--cycle;

\addplot[area legend, line width=0.5pt, draw=black!25.0!cyan!75, fill=black!25.0!cyan!75, fill opacity=0.5, forget plot]
table[row sep=crcr] {%
x	y\\
-6.41344086079725	2.08448587193603\\
-6.41344086079725	2.48448587193603\\
-6.61344086079725	2.48448587193603\\
-6.61344086079725	2.08448587193603\\
}--cycle;

\addplot[area legend, line width=0.5pt, draw=black!25.0!cyan!75, fill=black!25.0!cyan!75, fill opacity=0.5, forget plot]
table[row sep=crcr] {%
x	y\\
-3.61344086080059	2.59999999999666\\
-3.61344086080059	2.99999999999666\\
-3.81344086080059	2.99999999999666\\
-3.81344086080059	2.59999999999666\\
}--cycle;

\addplot [color=red, solid, line width=1.5pt, mark=x, mark options={solid, red}]
  table[row sep=crcr]{%
-8.74792673273328	-0.71551412806366\\
-6.86344086079725	2.08448587193603\\
-3.96344086080059	3.09999999999665\\
};
\addlegendentry{$\vec{X}^{3, 1}(\hat{\vec{x}}_0^{3, 1})$}
\addplot [color=green, dashed, mark=triangle, mark options={solid, green}]
  table[row sep=crcr]{%
-8.79792673273328	-0.41551412806366\\
-6.71344086079725	2.08448587193603\\
-3.81344086080059	2.89999999999665\\
};
\addlegendentry{$\hat{\vec{X}}^{3, 1}(\hat{\vec{x}}_0^{3, 1})$}
\addplot [color=blue, mark=diamond, mark options={blue}]
  table[row sep=crcr]{%
-8.79792673273328	-0.41551412806366\\
-6.51344086079725	2.28448587193603\\
-3.71344086080059	2.79999999999666\\
};
\addlegendentry{$\bar{\vec{X}}^{3, 1}(\hat{\vec{x}}_0^{3, 1})$}
\node[] at (axis cs: -5, 0.5) {$\mathbb{CS}^0 $};
\node[] at (axis cs: -4.0, 2.0) {$\mathbb{CS}^{2, 1} $};
\node[] at (axis cs: -8.7, 3.25) {$\mathbb{X}$};
\node[] at (axis cs: -8.7, 2.5) {$\textcolor{gray}{\bar{\mathbb{X}}}$};
\node[] at (axis cs: -7.5, -0.75) {$\textcolor{black!25.0!cyan!75}{\hat{\vec{x}}_0^{3, 1} \oplus \set{S}_\infty}$};
\end{axis}

\end{tikzpicture}%
	\caption{Unsuccessful trajectory of the certainty equivalent RL-MPC at iteration $j = 3$ of the first run. The state constraint is violated at the second time step of the iteration. } 
	\label{fig:violating_constraints}
\end{figure}

\subsection{Quadrotor Experiments}
We run a one-dimensional stabilization task on a Crazyflie~2.0 quadrotor.  
Our goal is to design reference signals with the proposed approach such that the quadrotor is driven to a neighborhood around the origin without violating the state and input constraints.
The quadrotor tracks the reference signals using a PD position controller. 
The position of the quadrotor is measured by a Vicon system at 300Hz. 
The quadrotor dynamics are nonlinear. However, by using the underlying PD controller and operating at low velocities, we can approximate the closed-loop dynamics of the quadrotor as a decoupled uncertain LTI system in the $x$-, $y$-, $z$-direction. In the experiment, we focus on the dynamics along the horizontal $x$-direction for illustration purposes and keep the other directions constant. 

We identify a linear system by applying a sinusoidal reference trajectory with time-varying frequency and amplitude. 
We consider the state $\vec{x} = \begin{pmatrix}
x_1 & x_2
\end{pmatrix}^\intercal \in \R^2$, where $x_1$ is the quadrotor's position and $x_2$ is its velocity in $x$ direction. The control input $\vec{u} \in \R$ is the position reference supplied to the PD controller. 
Using MATLAB's System Identification Toolbox we identified the system at 5Hz as
\begin{align*}
\begin{split}
\vec{x}_{k + 1} &=\begin{pmatrix}
0.8061 & 0.1090 \\
-1.4089 & 0.1419 
\end{pmatrix}\vec{x}_{k} + \begin{pmatrix}
0.1507 \\
1.5491
\end{pmatrix}\vec{u}_{k} + \vec{w}_k \,,\\
\vec{y}_{k} &=\begin{pmatrix}
1 & 0
\end{pmatrix}\vec{x}_{k} + \vec{v}_{k} \,.
\end{split}
\end{align*}
The state constraints are given by the area of the experimental space and the maximum velocities of the quadrotor $\set{X} = \{ \SI{-2.0}{\meter} \leq x_1 \leq \SI{2.0}{\meter}, \SI{-2.0}{\frac{\meter}{\second}} \leq x_2 \leq \SI{2.0}{\frac{\meter}{\second}} \}$. The control input set is the input space covered by the data used in the system identification: $\set{U} = \{ u \in \R : \SI{- 1.016}{\meter}  \leq u \leq \SI{ 0.9854}{\meter} \}$. 
The maximum and minimum errors from applying the same control inputs from the sinusoidal trajectory to the linear model compared to the true system are used as the system's process noise: $\set{W} = \{ \vec{w} \in \R^2 : \SI{-0.0233}{\meter} \leq w_1 \leq \SI{0.0201}{\meter}, \SI{-0.2049}{\frac{\meter}{\second}} \leq w_2 \leq \SI{0.1703}{\frac{\meter}{\second}} \}$. The measurement noise is identified by measuring the error when the quadrotor is on a fixed position on the ground and is given by $\set{V} = \{ v \in \R : \lvert v \rvert \leq \SI{0.001}{\meter} \}$. We pick $\vec{Q} = \mathrm{diag}(10, 100)$, $\vec{R} = 10$. The matrix $\vec{P}$ is obtained from the associated unconstrained LQR problem. We select the controller and observer gain as $\vec{K} = \vec{K}_f = \begin{pmatrix}
0.6568 & -0.1091
\end{pmatrix}$ and $\vec{L} = \begin{pmatrix}
-0.8608 & 1.4055
\end{pmatrix}^\intercal$, respectively. The RPI sets are determined through fixed-point iterations until~\eqref{eq:e_inf},~\eqref{eq:xi_inf},~and~\eqref{eq:z_inf} are satisfied.  We initialize the RLO-MPC with a trajectory from a robust MPC~\cite{Mayne2006}, but optimize for a different cost function with $\vec{Q}_{\mathrm{init}} = 10 \vec{Q}$, $\vec{R}_{\mathrm{init}} = 1$. The MPC horizon is $H = H_{\mathrm{init}} = 10$. 

We execute five separate runs with four iterations each. Each iteration is executed for $\SI{20}{\second}$.
Although we command the quadrotor to start from $\hat{\vec{x}}_0 = \begin{pmatrix} 0.7 & 0.0
\end{pmatrix}^\intercal$ in each iteration, precisely controlling the quadrotor to this initial condition is challenging.
Instead, we consider a set of initial conditions $\set{E}_0$. This initial state set is shown in magenta in Figure~\ref{fig:experiment_closed_loop}. 
To avoid delays in the computation of the control signal, we create the convex safe set and terminal cost function only from the last two iterations. 
The resulting closed-loop trajectory for the last iteration of the first run is shown in Figure~\ref{fig:experiment_closed_loop}. The true state $\vec{x}$ of the quadrotor is estimated offline using Vicon position measurements.
Figure~\ref{fig:experiment_closed_loop} shows that the measurement error is contained inside $\set{E}_\infty$ (shown for the first five time steps). Here, this set already contains all other errors as well. Since $\set{E}_\infty \subseteq \begin{pmatrix}
\vec{I} & \vec{I}
\end{pmatrix} \set{Z}_\infty$, we conclude that our modeling assumptions hold. The proposed RLO-MPC successfully controls the true state to a neighborhood of the origin without violating any state or input constraints. 
\begin{figure}[tb]
	\centering
	\input{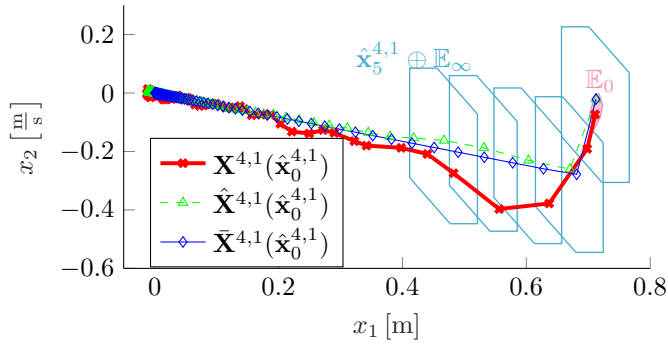}
	\caption{Closed-loop trajectory of the proposed RLO-MPC at iteration $j = 4$ of the first run on the quadrotor. As desired, the true state is controlled to a neighborhood of the origin despite modeling errors. Our modeling assumptions hold, as the true state $\vec{x}_k \in \hat{\vec{x}}_k \oplus \set{E}_\infty$ for all~$k$. The set $\set{E}_\infty$ at the state estimates for the first five steps is shown in light blue.} 
	\label{fig:experiment_closed_loop}
\end{figure}

For evaluating the nominal performance of the controller, we pick $T = 2$. Then we can determine the upper bounds on the nominal iteration cost, which is shown in Table~\ref{tab:experiment_upper_bound}. The minimum of each column is bold.
The upper bound $\min_r \bar{V}_{0 \to N}^{j, r}(\set{E}_0)$ is greater than the maximum of each iteration. Furthermore, the upper bound improves over the iterations. 
The upper bound can be improved by running additional iterations or by shrinking the set $\set{E}_0$. 

\begin{table}[tb]
	\centering
	\caption{The total nominal iteration cost $V^{j, r}_{0 \to N}(\hat{\vec{x}}_0^{j, r})$ and the derived upper bound $\bar{V}^j_{0 \to N}(\set{E}_0)$ for iterations $j \in \{0, \dots, 4\}$ and runs $r \in \{1, \dots, 5\}$ of the quadrotor experiments.
	}
	\begin{tabular}{lcccc}
		\toprule
		& \multicolumn{3}{c}{$V^{j, r}_{0 \to N}(\hat{\vec{x}}_0^{j, r})$} & \\
		j & $\min_r$ & $\max_r$ & $\mathrm{median}_r$ & $\min_r \bar{V}_{0 \to N}^{j, r}(\set{E}_0)$ \\
		\midrule
		0 & 106.555 & 112.746 & 108.724 &  -\\
		1 & 101.681 & 106.639 & \textbf{104.026} & 112.013\\
		2 & \textbf{100.808} & \textbf{106.261} & 104.103 & 107.692\\
		3 & 102.076 & 106.738 & 104.352 & 107.686\\
		4 & 104.776 & 106.338 & 105.491 & \textbf{107.677}\\
		\bottomrule
	\end{tabular}
\label{tab:experiment_upper_bound}
\end{table}
\balance 
\section{CONCLUSIONS}
We proposed RLO-MPC for performing iterative tasks. Our approach applies to LTI systems with uncertain observations and dynamics. 
We gave theoretical guarantees with regards to recursive feasibility, stability and performance of the closed-loop system. 
We relaxed the impractical assumption of identical initial conditions over all iterations and instead considered initial conditions from a polytopic set.
We showed that the worst-case achievable performance of our proposed controller is bounded. 
Furthermore, the derived upper bound is non-increasing and is reduced over iterations.
We successfully validated our proposed RLO-MPC in simulation and through experiments on a quadrotor. 
In simulation we compared our approach to a certainty equivalent approach, which does not consider state estimation errors. The certainty equivalent approach violated constraints for every run, while the proposed RLO-MPC always satisfied the constraints. This emphasizes the importance of accounting for uncertainty in the state estimation. 

\addtolength{\textheight}{-12cm}   



%
%

\bibliographystyle{IEEEtranS}
\bibliography{refs}

\end{document}